\pdfoutput=1
\documentclass{article}

\usepackage[T1]{fontenc}
\usepackage{graphicx}
\usepackage{graphicx} 
\usepackage[all,pdf]{xy}
\usepackage{amsmath}
\usepackage{amssymb}
\usepackage{theorem}
\usepackage{booktabs}
\usepackage{stmaryrd}

\newcommand{\sa}{\dagger} 
\newcommand{\sai}{\ddagger} 
\newcommand{\pre}{\sf{pre}}
\newcommand{\safe}{\blacktriangle}
\newcommand{\ext}[1]{||#1||}
\newcommand{\pow}{\mathcal{P}}
\newcommand{\cl}{\mathit{cl}}
\newcommand{\onlyone}{\mathit{onlyone}}
\newcommand{\cignorant}{\mathit{cignorant}}
\newcommand{\bbullet}{\circ}

\newcommand{\lang}{\mathcal{L}}
\newcommand{\safelang}{\ensuremath{\lang_\safe}}
\newcommand{\sailang}{\ensuremath{\lang_\sai}}
\newcommand{\salang}{\ensuremath{\lang_\sa}}
\newcommand{\saifelang}{\ensuremath{\lang_{\sai\safe}}}
\newcommand{\ellang}{\ensuremath{\lang_K}}
\newcommand{\pallang}{\ensuremath{\lang_!}}
\newcommand{\amllang}{\ensuremath{\lang_\otimes}}

\newcommand{\safesys}{\ensuremath{\mathbf{S\safe}}}

\usepackage{soul}
\usepackage{todonotes}

\newtheorem{proposition}{Proposition}

\newtheorem{definition}{Definition}
\newtheorem{lemma}{Lemma}
\newtheorem{theorem}{Theorem}
\newtheorem{corollary}{Corollary}
\newenvironment{proof}{\paragraph{Proof:}}{\hfill$\square$}

\begin{document}
\title{Anonymous Public Announcements}

\author{Thomas {\AA}gotnes\\
  University of Bergen, Bergen, Norway\\
  Shanxi University, Taiyuan, China\\
  \texttt{thomas.agotnes@uib.no}
\and
Rustam Galimullin\\University of Bergen, Bergen, Norway\\
\texttt{Rustam.Galimullin@uib.no}
\and
Ken Satoh\\Center for Juris-Informatics, Tokyo, Japan\\
\texttt{ksatoh@nii.ac.jp}
\and
Satoshi Tojo\\Asia University, Tokyo, Japan\\
\texttt{tojo.satoshi@gmail.com}
}
\date{}

\maketitle             

\begin{abstract}
 We formalise the notion of an \emph{anonymous public announcement} in the tradition of public announcement logic. Such announcements can be seen as in-between a public announcement from ``the outside" (an announcement of $\phi$) and a public announcement by one of the agents (an announcement of $K_a\phi$): we get more information than just $\phi$, but not (necessarily) about exactly who made it. Even if such an announcement is \emph{prima facie} anonymous, depending on the background knowledge of the agents it might reveal the identity of the announcer: if I post something on a message board, the information might reveal who I am even if I don't sign my name. Furthermore, like in the Russian Cards puzzle, if we assume that the announcer's \emph{intention} was to stay anonymous, that in fact might reveal more information. In this paper we first look at the case when no assumption about intentions are made, in which case the logic with an anonymous public announcement operator is reducible to epistemic logic. We then look at the case when we assume common knowledge of the intention to stay anonymous, which is both more complex and more interesting: in several ways it boils down to the notion of a ``safe" announcement (again, similarly to Russian Cards). Main results include formal expressivity results and axiomatic completeness for key logical languages.
\end{abstract}

\section{Introduction}

Taken at the face value, the title of this paper seems to be an oxymoron. Indeed, if ``public announcement" is taken literally, as in an agent \emph{saying} something in front of everyone else, it will not be anonymous. However, anonymous public communication is almost ubiquitous in our day-to-day lives. Think of posts on social media and message boards done under a username instead of a real name of the poster. Or an anonymous letter to an editor of a news outlet. Other examples include anonymous emails, transactions on a blockchain, whistle blower reports, and even cultural artifacts created anonymously under an alias, like Elena Ferrante and MF DOOM.

The type of anonymity we are interested here focuses on \textit{action anonymity}, i.e. the inability of an attacker to identify who performed a given action (also sometimes referred to as \textit{unlinkability} in the literature \cite{anon_terminology}). This is in contrast to \textit{data anonymity}, i.e. the inability of an attacker to know the identity of a subject in an anonymised database, e.g. in medical records (see, e.g., \cite{domingo-ferrer16}). One of the standard requirements to both types of anonymity is that they satisfy $k$-anonymity \cite{sweeney02}, which intuitively means that a data record or an action cannot be distinguished from at least $k-1$ other records or actions. It is clear that in the case of public communication, we should have at least 3-anonymity. Indeed, if a public announcements is so specific that it could be done only by two agents (and these two agents know it), then the non-announcing agent would be able to deduce the identity of the announcer. 

In this paper we formalise anonymous public announcements inspired by \textit{public announcement logic} (PAL) \cite{plaza89}, an extension of multi-agent \textit{epistemic logic} with constructs of the form $[\phi!]\psi$ intuitively meaning that after $\phi$ is truthfully announced, $\psi$ is true. In PAL the announced formula $\phi$ does not have to actually be known by any agent in the system -- the identity of the announcer is left out of the picture. If the announcer indeed is one of the agents $a$ in the system, the announcement in fact contains \emph{more} information: in that case it would be modelled by the announcement $K_a\phi!$. In this paper we formalise \emph{anonymous} public announcements, conceptually somewhere in-between $\phi!$ and $K_a\phi!$ -- we get \emph{more} information than just $\phi$ but \emph{less} than $K_a\phi$ for a specific agent $a$.  We thus first introduce modalities of the form $[\psi \sa]$ where $[\psi \sa] \phi$ that are read as ``after $\psi$ is \textit{pseudo-anonymously} announced, $\phi$ is true no matter who the announcer was". Such announcements are pseudo-anonymous because we do not assume that the announcer necessarily \emph{intends} to stay anonymous. Indeed, the reader may have noticed that not all the examples of anonymous public communication we mentioned above are guaranteed to preserve the identity of the announcer.  For example, if the announced information is known only to two agents (and the non-announcing agent knows it), then the identity of the announcer will not be a secret to the non-announcing agent. In the literature, such a scenario is called ``background knowledge attack" \cite{machanavajjhala07}.  Thus, we also formalise \textit{intentional} (or \textit{safe}) \textit{anonymous announcements} by introducing constructs $[\psi \sai]\phi$ meaning ``after the safe announcement of $\psi$, $\phi$ is true no matter who the announcer was". In order to capture safety, we employ the idea of (at least) 3-anonymity, and agents intending to stay anonymous should know that their announcement is safe. A bit more formally, we introduce the safety modality $\safe \psi$ which, similarly to common knowledge, is defined via a fixpoint. Hence, $\safe \psi$ means that ``$\psi$ is true, and that there is a group of three agents that know that $\psi$, and there is a group of three agents that know that there is a group of three agents that know that $\psi$ and so on". In other words, safe anonymous announcements guarantee the anonymity of announcers.

Reasoning about anonymity based on (variants of) epistemic logic has been studied in \cite{syverson99,halpernO05}, with $k$-anonymity being discussed in \cite{halpernO05}. Building on the runs-and-systems approach of \cite{halpernO05}, further extensions focusing on privacy and onymity \cite{tsukada09,tsukada6} and electronic voting \cite{mano10}. A knowledge-based approach to data anonymity was presented in \cite{jiang25}. Other logical approach to anonymity and privacy include \cite{schneiderS96,hughesS04,ye07,barthe14,naumovO23}.
The themes of anonymity and privacy are also related to the research on secrets in multi-agent systems (see, e.g., \cite{halpernO08,moreN11,moreN11b,xiongA23}). None of these approaches model anonymous \emph{announcements}.
Finally, an approach to anonymity based on \textit{dynamic epistemic logic} (DEL) \cite{vDvdHK2007} was presented in \cite{VANEIJCK2007159}, where the authors consider scenarios of secret communication between agents in a system. Such secret communication is captured by private announcements, as special type of \textit{action models} \cite{DBLP:journals/synthese/BaltagM04}. In this setting, verifying whether secret communication remained secret boils down to model checking an epistemic formula in the resulting updated model, i.e. the model that is obtained after an application of a dynamic operator.  

While our approach is also DEL-based, in our work we focus on anonymous \textit{public} communication, as opposed to \textit{private} communication. Moreover, we tackle the issue of intentional anonymity, where the announcing agent is guaranteed to stay anonymous. To this end, we introduce a novel safety modality that, to the best of our knowledge, has never been considered in the literature before.

We start out by introducing the machinery of EL and DEL in Section \ref{sec:background}. In Section \ref{sec:unsafe}, we study pseudo-anonymous announcements, and in particular show that while the resulting logic is as expressive as EL, standard and pseudo-anonymous announcements are update incomparable. After that, in Section \ref{sec:intentions}, we formalise intentional anonymous announcements using the safety modality, as well as show that, perhaps surprisingly, safety is all one needs: logics with safe announcements are as expressive as EL extended with the safety modality. Moreover, we provide a sound an complete axiomatisation of the latter. Finally, in Section \ref{sec:discussion} we discuss further research directions.

\section{Background}
\label{sec:background}
Let $N$ be a finite set of \textit{agents}, and $P$ be a countable set of \textit{propositional variables}. 

All logics that we are dealing with in this paper are interpreted on epistemic models.

\begin{definition}
    An \emph{(epistemic) model} is a triple $M= (S, \sim, V)$, where $S$ is a non-empty set of states, $\sim: N \to 2^{S \times S}$ is an equivalence relation for each $i\in N$, and $V:P \to 2^S$ is the valuation function. For $s \in S$, a pair $M,s$ is called a \emph{pointed (epistemic) model}. 
\end{definition}

Before we define some basic logical languages, we need to introduce a special kind of models that will be used in syntax.

\begin{definition}
    Let $\lang$ be a language defined over the signature $\langle N, P \rangle$. An \emph{action model} is a triple $\mathsf{M = (S, \sim, pre)}$, where $\mathsf{S}$ is a non-empty set of states, $\sim: N \to 2^{\mathsf{S} \times \mathsf{S}}$ is an equivalent relation for each $i\in N$, and $\mathsf{pre}: \mathsf{S}\to\lang$ is the precondition function. 
    For $\mathsf{s} \in \mathsf{S}$, we will call a pair $\mathsf{M},\mathsf{s}$ a \emph{pointed action model}.
\end{definition}

\begin{definition}
Languages of \emph{epistemic logic} (\ellang), \emph{public announcement logic} (\pallang), and \emph{action model logic} (\amllang) are defined by the following BNFs:
       \begin{alignat*}{3}
        &\ellang &&\thinspace \ni && \enspace \phi ::= p \mid \neg \phi \mid (\phi \land \phi) \mid K_i \phi\\
         &\pallang &&\thinspace \ni && \enspace \phi ::= p \mid \neg \phi \mid (\phi \land \phi) \mid K_i \phi \mid [\phi !]\phi\\
          &\amllang &&\thinspace \ni && \enspace \phi ::= p \mid \neg \phi \mid (\phi \land \phi) \mid K_i \phi \mid [\pi] \phi\\
          & && &&\enspace \pi:=(\mathsf{M,s}) \mid \pi \cup \pi
\end{alignat*}
where $p \in P$, $i \in N$, and $(\mathsf{M,s})$ is a pointed action model with a finite set of states $\mathsf{S}$, and such that for all $\mathsf{s \in S}$, precondition $\mathsf{pre(s)}$ is some $\phi \in \amllang$ that was constructed in a previous stage of the inductively defined hierarchy.
\end{definition}

Having defined models and languages, we are now ready to provide the definition of the semantics of the aforementioned logics.

\begin{definition}
    Let $M,s = (S, \sim, V)$ be a model, $p \in P$, $i \in N$, and $\mathsf{(M,s)}$ and be an action model.
\begin{center}
$\begin{array}{lll}
	M,s \models p & \text{iff} & s \in V(p)\\
	M,s \models \lnot \phi & \text{iff} & M,s \not \models \phi\\
	M,s \models \phi \land \psi & \text{iff} & M,s \models \phi \text{ and } M,s \models \psi\\
	M,s \models K_i \phi & \text{iff} & M,t \models \phi \text{ for all  } t \in S \text{ such that } s \sim_i t \\
    M,s \models [\psi!] \phi & \text{iff} & M,s \models \psi \text{ implies } M,s^{\psi!} \models \phi\\
    M,s \models [\mathsf{M,s}] \phi  & \text{iff} & M,s \models \mathsf{pre(s)} \text{ implies } (M \otimes \mathsf{M}, (s,\mathsf{s})) \models \phi \\
    M,s \models [\pi \cup \rho] \phi  & \text{iff} &  M,s \models [\pi] \phi \text{ and } M,s \models [\rho] \phi 
\end{array}$
\end{center}
Given $M$, we will write $\llbracket \phi \rrbracket_M$ for the set $\{s \in S \mid M,s \models \phi\}$.

The updated model $M^{\phi!}$ is $(S^{\phi!}, \sim^{\phi!}, V^{\phi!})$, where $S^{\phi!} = \llbracket \phi \rrbracket_M$, $\sim_i^{\phi!} = \sim_i \cap (S^{\phi!} \times S^{\phi!})$ for all $i \in N$, and $V^{\phi!} (p) = V(p) \cap \llbracket \phi \rrbracket_M$ for all $p \in P$.

Updated model $M \otimes \mathsf{M}$ is $(S', \sim',V')$, where $S' = \{(s, \mathsf{s}) \mid s\in S, \mathsf{s \in S}, M,s \models \mathsf{pre(s)}\}$, $(s, \mathsf{s})R'_i(t, \mathsf{t})$ iff $s\sim_i t$ and $\mathsf{s}\sim_i\mathsf{t}$, and $(s, \mathsf{s})\in V'(p)$ iff $s\in V(p)$.
 
We call a formula $\phi$ \emph{valid}, or a \emph{validity}, if for all $M,s$ it holds that $M,s \models \phi$.
\end{definition}

\begin{definition}
Let $M^1 = (S^1, \sim^1, V^1)$ and $M^2 = (S^2, \sim^2, V^2)$ be two epistemic models. We say that $M^1$ and $M^2$ are \emph{bisimilar} (denoted $M^1 \leftrightarrows M^2$) if there is a non-empty relation $Z \subseteq S^1 \times S^2$, called \emph{bisimulation}, such that for all $sZt$, the following conditions are satisfied:
\begin{description}
\item[\textit{Atoms}] for all $p \in P$: $s \in V^1(p)$ if and only if $t \in V^2(p)$,
\item[\textit{Forth}] for all $i \in N$ and  $u \in S^1$ s.t. $s\sim^1_i u$, there is a $v \in S^2$ s.t. $t\sim^2_i v$ and $uZv$,
\item[\textit{Back}] for all $i \in N$ and $v \in S^2$ s.t $t\sim^2_i v$, there is a $u \in S^1$ s.t. $s\sim^1_i u$ and $uZv$. 
\end{description} 
We say that $M^1,s$ and $M^2,t$ are bisimilar and denote this by ${M^1,s \leftrightarrows M^2,t}$ if there is a bisimulation linking states $s$ and $t$.
\end{definition}

It is a standard result that $M^1,s \leftrightarrows M^2,t$ implies $M^1,s \models \phi$ if and only if $M^2,t \models \phi$ for $\phi \in \ellang$(see, e.g, \cite{goranko07}). In other words, epistemic logic is bisimulation invariant over epistemic models. We also have that \pallang\ and \amllang\ are bisimulation invariant over epistemic models \cite[Chapter 5]{vDvdHK2007}.

\begin{definition}
    Let $\lang_1$ and $\lang_2$ be two languages, and let $\phi \in \lang_1$ and $\psi \in \lang_2$. 
We say that $\phi$ and $\psi$ are \emph{equivalent}, when for all models $M_s$: $M_s \models \phi$ if and only if $M_s \models \psi$. 

If for every $\phi \in \lang_1$ there is an equivalent $\psi \in \lang_2$, we write $\lang_1 \preccurlyeq \lang_2$ and say that $\lang_2$ is \emph{at least as expressive as} $\lang_1$. We write $\lang_1 \prec \lang_2$ iff $\lang_1 \preccurlyeq \lang_2$ and $\lang_2 \not \preccurlyeq \lang_1$, and we say that $\lang_2$ is \emph{strictly more expressive than} $\lang_1$. Finally, if $\lang_1 \preccurlyeq \lang_2$ and $\lang_2 \preccurlyeq \lang_1$, we say that $\lang_1$ and $\lang_2$ are \emph{equally expressive} and write $\lang_1 \approx \lang_2$.
\end{definition}

The classic result in DEL is that $\pallang \approx \amllang \approx \ellang$ \cite[Chapter 8]{vDvdHK2007}, i.e. that any formula with public announcements or action models can be reduced to an equivalent formula of epistemic logic. 

While \pallang and \amllang are equally expressive, in some sense action models are more powerful than public announcements. Indeed, a public announcement of $\psi$ can be modelled by the single-state action model $\mathsf{M}^\psi = (\{\mathsf{s}, \{\mathsf{s} \sim_i \mathsf{s}| i \in N\}, \mathsf{pre}(\mathsf{s}) = \psi\})$. At the same time, not every action model can be modelled by a public announcement: action models may result in the bigger updated models, while public announcements result in the updated models of the size of at most the size of the original one. We can capture this disctinction formally.

\begin{definition}
    Let $\lang_1$ and $\lang_2$ be two languages with dynamic operators.
If for every update $[\alpha]$ of $\lang_1$ there is an update $[\beta]$ of $\lang_2$ such that for all $M,s$, the update of $M,s$ with $\alpha$ is bisimilar to the update of $M,s$ with $\beta$, 
then we say that $\lang_2$ is \emph{at least as update expressive as} $\lang_1$ and write $\lang_2 \preccurlyeq^U \lang_1$. We write $\lang_1 \prec^U \lang_2$ iff $\lang_1 \preccurlyeq^U\lang_2$ and $\lang_2 \not \preccurlyeq^U \lang_1$, and we say that $\lang_1$ is \emph{strictly less update expressive than} $\lang_2$.If $\lang_1 \preccurlyeq^U \lang_2$ and $\lang_2 \preccurlyeq^U \lang_1$, we say that $\lang_1$ and $\lang_2$ are \emph{equally update expressive} and write $\lang_1 \approx^U \lang_2$. Finally, if $\lang_1 \not \preccurlyeq^U \lang_2$ and $\lang_2 \not \preccurlyeq^U \lang_1$, we say that $\lang_1$ and $\lang_2$ are \emph{incomparable in update expressivity} and write $\lang_1 \not \approx^U \lang_2$.
\end{definition}

See \cite{vanditmarsch20} for a more thorough discussion on update expressivity as well as an overview of update expressivity results for some DELs, including $\pallang \prec^U \amllang$.

\section{Unintentional Anonymity}
\label{sec:unsafe}

We are interested in anonymous public communication that we encounter in anonymous P2P networks, Internet forums, some blockchains, and so on. We can also think about the event of ``all the agents come into the classroom, and see that someone (one of the agents) has written $\phi$ on the blackboard". 
This announcement is fundamentally anonymous in the sense that the announcers' identity is not explicitly revealed, but it might be \emph{implicitly} revealed in certain situations depending on the agents' background knowledge. A trivial example is if you write something on the blackboard that I know that only you know. We thus call such events \emph{pseudo-anonymous announcements}. 

Unlike a standard public announcement of $\phi$ by $a$, a pseudo-anonymous announcement of $\phi$ by $a$ is no longer modeled by a public announcement of the formula $K_a\phi$, since the announcement is not signed by $a$. On the other hand, it is also not modeled by a public announcement of $\phi$, since the event also contains information that \emph{some} agent knows $\phi$. The example ``someone has written $\phi$ on the blackboard" hints to a public announcement of ``someone knows $\phi$", i.e., $\bigvee_{i \in N}K_i\phi$. 

However, updating the model by removing states where $\bigvee_{i \in N}K_i\phi$ is false also does not work, because, e.g., it does not take into account the fact that the announcer knows who the announcer is.  Indeed, this also shows that pseudo-anonymous announcements \emph{cannot be modeled by public announcements at all, because the former are not deterministic}\footnote{We will make this observation more precise in Proposition \ref{prop:updexp}.}:  it might be that $a$ announced $\phi$, and it might be that $b$ announced $\phi$, and these are two different model updates that need to be considered separately.

We will do exactly that, by first defining the updated model and then defining the semantics a new dynamic modality $[\phi\sa]\psi$ in terms of that. Similar to the standard public announcement operator $[\phi!]$, 
$[\phi\sa]\psi$
is intended to mean that $\psi$ is true after $\phi$ is pseudo-anonymously announced.  There is a very direct relationship to action models, which we make explicit in Section \ref{sec:actionmodels} below. 

\begin{definition}
    The \emph{update of epistemic model} $M = (S,\sim,V)$ by the pseudo-anonymous announcement of $\phi$ is the epistemic model $M^{\phi\sa} = (S',\sim',V')$ where:
    \begin{itemize}
        \item $S' = \{(s,a) \mid  s \in S, a \in N \text{ and } M,s \models K_a \phi\}$
        \item $(s,a) \sim'_c (t,b)$ iff $s \sim_c t$ and $a = c$ iff $b = c$ 
        \item $V'(s,a) = V(s)$
    \end{itemize}
\end{definition}

Intuitively, in the updated model it is common knowledge that someone had (truthfully) announced/posted/written on the blackboard that $\phi$. A state $(s,a)$ in the updated model corresponds to that someone being $a$, in state $s$ of the original model. The precondition is that $a$ knows $\phi$, corresponding to the assumption that anonymous public announcements are truthful. An agent can only discern between a situation $(s,a)$ where $a$ made the pseudo-anonymous announcement and a situation $(t,b)$ where a different agent $b$ did, if she could already discern between $s$ and $t$ before the announcement or she is neither $a$ nor $b$.

\begin{definition}
The language of \emph{pseudo-anonymous public announcement logic} \salang is recursively defined by the BNF 
       \begin{alignat*}{3}
           \phi ::= p \mid \neg \phi \mid (\phi \land \phi) \mid K_i \phi \mid [\phi\sa]\phi
\end{alignat*}
Intuitively, $[\phi\sa]$ is the event ``$\phi$ is pseudo-anonymously announced". $[\phi\sa]\psi$ means that $\psi$ is necessarily true after that event, no matter who the announcer was. The dual is defined as $\langle\phi\sa\rangle = \lnot [\phi\sa]\lnot\psi$, and intuitively means that there is an agent who can pseudo-anonymously announce $\phi$ such that $\psi$ will be true afterwards. 
\begin{center}
$\begin{array}{lll}
    M,s \models [\phi\sa]\psi & \text{iff} & \forall a \in N: M,s \models K_a\phi \text{ implies } M^{\phi\sa}, (s,a) \models \psi
\end{array}$
\end{center}
\end{definition}

For an example, see Figure \ref{fig:update}. Let $M$ be the epistemic model in the middle. The figure also show three different updates related to announcements about $p$: someone outside the system announced $p$ (top); it was announced that someone inside the system knows $p$ (bottom right); someone inside the system pseudo-anonymously announced that they know $p$ (bottom left).

\begin{figure}[t!]
    \centering
\[\xymatrix{
&&\bullet^{p,q}_u\ar@{..}[r]^{ab}&\bullet^{p,r}_s\ar@{..}[d]^{b}\ar@{..}[r]^b &\bullet^{p}_v\\
&&&\bullet^{p,o}_x\ar@{..}[ul]^b\ar@{..}[ur]^b\\
&&&\uparrow p!\\ 
&\bullet^{\neg p}_t\ar@{..}[r]^{bc}&\bullet^{p,q}_u\ar@{..}[r]^{ab}&\bullet^{p,r}_s\ar@{..}[d]^{b}\ar@{..}[r]^{b}&\bullet^{p}_v\ar@{..}[r]^{ab}&\bullet^{\neg p}_w\\
&&&\bullet^{p,o}_x\ar@{..}[ull]^{ab}\ar@{..}[ul]^b\ar@{..}[ur]^b\ar@{..}[urr]^c\\
&&\swarrow p\sa && \searrow \bigvee_{i \in N} K_i p!\\ 
&\bullet^{p,r}_{(s,a)}\ar@{..}[dd]^{b}\ar@{..}[dr]^{b}\\
\bullet^{p,q}_{( u,a)}\ar@{..}[rr]_{b}\ar@{..}[ur]^{ab}\ar@{..}[dr]^{b}&&\bullet^{p}_{(v,c)}&& \bullet^{p,q}_u\ar@{..}[r]^{ab}&\bullet^{p,r}_s\ar@{..}[r]^b&\bullet^{p}_v\\
&\bullet^{p,r}_{(s,c)}\ar@{..}[ur]^{b}
}\]
    \caption{Three-agent epistemic model (middle) and its update after the event ``somebody pseudo-anonymously said $p$" (bottom left) as well as the update after the public announcement of ``somebody knows $\phi$'' (bottom right) and the update after the public announcement ``$\phi$ is true" (top). (Assume transitive closure, not all edges are shown.)}
    \label{fig:update}
\end{figure}
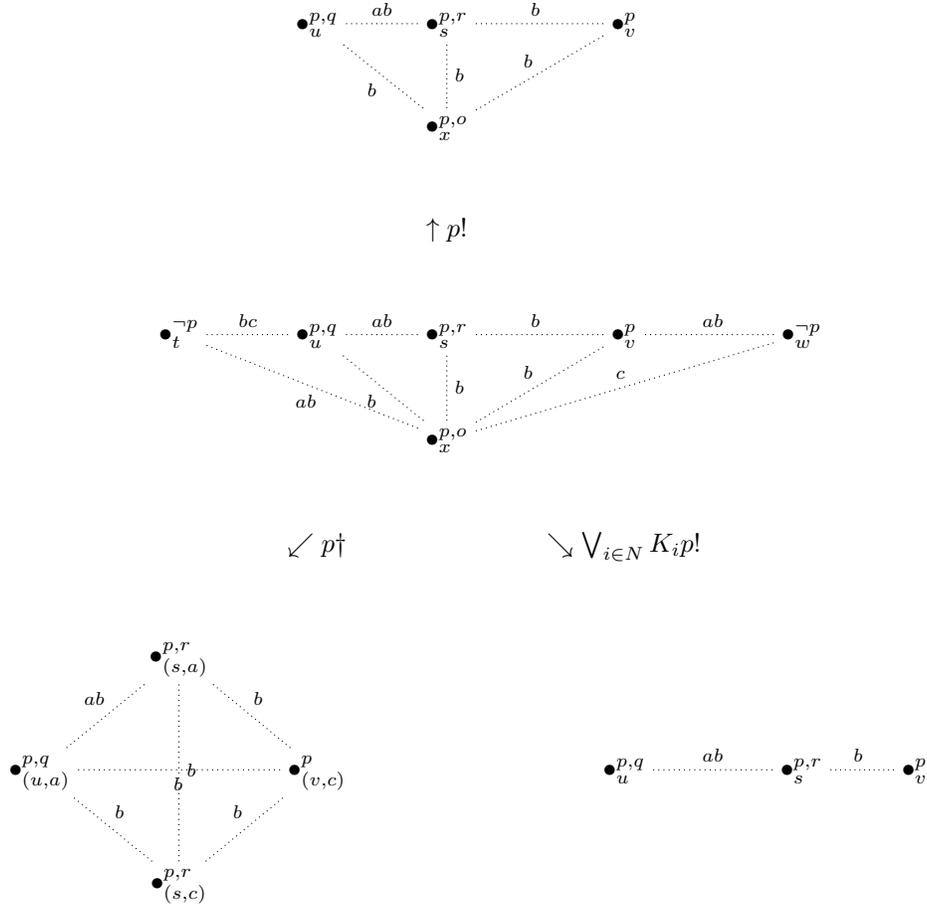

Note that in this example:
\begin{itemize}
    \item All the three updates are different.  
    \item $M,s \models [p\sa]\psi$ holds iff $\psi$ holds in both states ${(s,a)}$ and ${(s,c)}$ in the updated model to the bottom left in the figure.
    \item Thus we have for example that $M,s\models [p\sa]\hat{K}_b (r \wedge K_a \neg q)$: after the announcement $b$ considers it possible that the initial state was $s$ and that $a$ has learned that $q$ is false. This holds because $b$ considers it possible it was $c$ who announced $p$ (state $(s,c)$ in the updated model), in which case $a$ has learned that we were not in $u$. Note that this formula does not hold after the two other updates: $M,s \not\models [p!]\hat{K}_b (r \wedge K_a \neg q)$ and $M,s \not\models [\bigvee_{i \in N} K_i p!]\hat{K}_b (r \wedge K_a \neg q)$.
    \item The difference between the updates $\phi\sa$ and $\bigvee_{i \in N} K_i\phi!$ is related to the familiar ``de re"/``de dicto" distinction: in the latter case we update with ``in every state there exists someone who knows $\phi$", in the former with ``there exists someone who knows $\phi$ in every state". 
    \item We also have that $M,s \not\models [p\sa]K_a r$: if it was $a$ herself who made the pseudo-anonymous announcement (state $(s,a)$ in the updated model), she wouldn't learn that we were in $s$. 
\end{itemize}

Pseudo-anonymous announcements are not necessarily \emph{truly} anonymous, in all situations w.r.t. all agents, because depending on the background information of the agents they might reveal the identity of the announcer. 
Continuing with the example:
\begin{itemize}
    \item $a$'s pseudo-anonymous announcement of $p$ is truly anonymous w.r.t. $b$, in the epistemic state $M,s$. Indeed, $b$ considers it possible that it was $c$ who announced $p$. We see this explicitly in the updated model: $b$ considers two states possible where $a$ announced $p$ and two where $c$ did it.
    \item $a$' pseudo-anonymous announcement of $p$ is not truly anonymous w.r.t. $c$, in the epistemic state $M,s$. Intuitively, this is because $c$ knows that the only other agent besides herself who can announce $p$ is $a$. Formally, this can be seen from the updated model: in the state ${(s,a)}$, $c$ knows that it was $a$ who made the announcement. 
\end{itemize}

\subsection{Relationship to Action Models}
\label{sec:actionmodels}

As mentioned above, there is a close relationship with 
action model logic. In fact, the semantics of $\phi\sa$ given above is equivalent to a certain class of action models, \emph{pseudo-anonymous action models}, as we now show.

\begin{definition}
    The \emph{pseudo-anonymous action model} for $N$ agents and formula $\phi$ is the action model ${\sf M^N_\phi = (S,\sim,\pre)}$ where
    \begin{itemize}
        \item ${\sf S} = N$
        \item ${\sf a \sim_b c}$ iff $({\sf a} = {\sf c} \Leftrightarrow {\sf b} = {\sf c})$
        \item ${\sf \pre(a)} = K_a \phi$
    \end{itemize}
\end{definition}

Intuitively, the pointed pseudo-anonymous action model $({\sf M^N_\phi,a})$ corresponds to $a$ writing $\phi$ on the blackboard. The precondition is that $a$ knows $\phi$. Agent $c$ cannot discern between events $\sf a$ and $\sf b$ if and only if either $\sf a=b$, or ${\sf a} \neq {\sf c}$ and ${\sf b} \neq {\sf c}$. The three-agent pseudo-anonymous event is illustrated in Figure \ref{fig:3agent}.

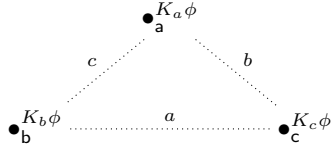
\begin{figure}
    \centering
\[\xymatrix{
    &\bullet^{K_a\phi}_{\sf a}\ar@{..}_c[dl]\ar@{..}^b[dr]\\
    \bullet^{K_b\phi}_{\sf b}\ar@{..}^a[rr]&&\bullet^{K_c\phi}_{\sf c}
}\]
    \caption{The three-agent pseudo-anonymous action model.}
    \label{fig:3agent}
\end{figure}

A pseudo-anonymous announcement now corresponds to the \emph{union} of events $({\sf M^N_\phi,a})$ for all agents $a$.  The proof of the following is straightforward by the definition of semantics. 
\begin{proposition}
\label{prop:anpalVSaml}
    For any pointed epistemic model $M,s$ and formulas $\phi$ and $\psi$,
    $M,s \models [\phi\sa]\psi \text{ iff } M,s \models \left[\sf \bigcup_{i \in N} (M^N_\phi,i)\right]\psi$.
\end{proposition}

From this we get that \salang is no more expressive than \ellang via the known fact that $\amllang \approx \ellang$.

\subsection{Update expressivity} Given that the logic of pseudo-anonymous announcements is as expressive as epistemic logic, we show that pseudo-anonymous announcements are still unique dynamic operators. To this end, we claim that $\salang \not \approx^U \pallang$ and $\salang \not \approx^U \amllang$.

\begin{proposition}
\label{prop:updexp}
    \salang and \pallang are update incomparable. 
\end{proposition}

\begin{proof}
Consider model $M$ in Figure \ref{fig:updexp} and public announcement $p!$. The result of the public announcement of $p!$ in $M$ is the model $M^{p!}$ consisting of the single state $s$ satisfying $p$.  
    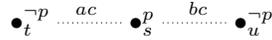
\begin{figure}[h!]
    \centering
\[\xymatrix{
\bullet^{\lnot p}_t\ar@{..}[r]^{ac}&\bullet^{p}_s \ar@{..}[r]^{bc} &\bullet^{\lnot p}_u
}\]
    \caption{Epistemic model $M$.}
    \label{fig:updexp}
\end{figure}

To see that there is no pseudo-anonymous announcement $\phi \sa$ that will result in a model bisimilar to $M^{p!}$, it is enough to notice that none of the agents in model $M$ has the ability to remove $\lnot p$-states due to precondition $K_i \phi$ for pseudo-anonymous announcements. 

To show that there are pseudo-anonymous announcements that cannot be captured by public announcements, we turn back to the model $M$ in the centre of Figure \ref{fig:update} and its update $M^{p \sa}$ depicted at bottom left of the same figure. 
Moreover, consider formula $\phi:= r \land \widehat{K}_a q \land \widehat{K}_b K_a (\lnot q \land r) \land K_b p$. This verbose formula is true in $M^{p\sa}, (s, a)$, and intuitively states that the current state satisfies $r$ (i.e. we are in an $s$-state), all reachable states satisfy $p$ (via conjunct $K_b p$ and the fact that $b$'s relation is universal), that there is an $a$-reachable $q$-state (conjunct $\widehat{K}_a q$), and that there is a state, where $a$ knows $\lnot q \land r$ (state $(s,c)$).

Now we will argue that no public announcement made in $M$ can result in an update satisfying $\phi$. First, since we have $K_b p$ as a conjunct in $\phi$, any update should preserve only $p$-states, i.e. a non-empty subset of $\{s,u,v,x\}$. State $s$ should be preserved as it is the only one satisfying $r$. Similarly, state $u$ should be in the updated model to satisfy $\widehat{K}_a q$. Now, the third conjunct of $\phi$ forces the existence of a state, where $a$ knows $\lnot q \land r$. The only state that satisfies $\lnot q \land r$ is $s$. However, since we have to preserve $u$, it would not hold in $s$ that $K_a (\lnot q \land r)$. It is easy to check that none of the remaining options of adding $x$ or $v$ will lead to the satisfaction of $K_a (\lnot q \land r)$ in $s$. Thus, there is no public announcement that can be made in $M$ such that the updated model exists and satisfies $\phi$. This implies that there is no public announcement update of $M$ that is bisimilar to $M^{p \sa}, (s, a)$.
\end{proof}

\begin{corollary}
\salang is strictly less update expressive than \amllang.
\end{corollary}

This follows from the facts that (1) each anonymous announcement can be modelled by the corresponding action model (Proposition \ref{prop:anpalVSaml}) and hence $\salang \preccurlyeq^U \amllang$; (2) each public announcement can be modelled by a single-state action model ($\pallang \not \preccurlyeq^U \salang$ from Proposition \ref{prop:updexp} and hence $\amllang \not \preccurlyeq^U \salang$).

\section{Intentional Anonymity}
\label{sec:intentions}

We now turn to \emph{intentional} anonymous announcements, where we assume that it is common knowledge that the announcement was intended to be anonymous. We want to capture that idea by strengthening the pre-condition on announcements, so that only ``safe" announcements are allowed. We will discuss in detail what ``safe" means, and it turns out to be quite subtle, but the intuition is that it is safe for an agent to make an announcement if somebody else could have made it. One of the subtleties is that the potential announcement by that ``somebody else" also should be safe, so the definition has a recursive (or even circular!) flavour.

\begin{figure}
    \centering
\[\xymatrix{\bullet^{p}_s\ar@{..}[r]^{b}&\bullet^{\neg p}&
\Rightarrow p\sa&
\bullet^{p}_{(s,\sf a)}\ar@{..}[d]^{b}\\
&&&\bullet^{p}_{(s,\sf c)}
}\]
    \caption{Three-agent epistemic model (top) and its update after a pseudo-anonymous announcement of $p$ (bottom). Accessibility for agents $a$ and $c$ is the identity relation. }
    \label{fig:trivial}
\end{figure}
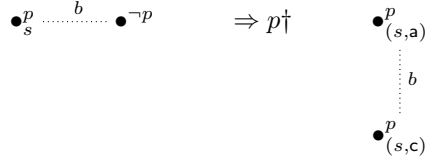

As a first simple example, consider the three-agent model in Figure \ref{fig:trivial}. We have that both $M,s \models K_a p$ and $M,s \models K_c p$ so both $a$ and $c$ can announce $p$, and we see in the updated model that in either case $b$ doesn't know who did it. However, since there are only two agents who could make the announcement, ``the other" agent learns who made it -- if $a$ makes the announcement, $c$ knows that $a$ did it, and vice versa. Thus, \emph{2-anonymity is not enough}. This is because we modeling anonymity \emph{within the system}, and not from the outside. So, we need at least \emph{3-anonymity}: an announcement is only safe if at least three agents can make it. In the model in the figure, no safe announcements of $p$ can be made.

Let's move on to the model $M$ in Figure \ref{fig:intentions1}, where also the updated model after the pseudo-anonymous announcement of $p$ is shown. In state $s$, all of 
$a$, $b$ and $c$ know $p$ and thus any one of them could have written it on the blackboard. The 
 updated model reflects that no-one except the announcer knows who did it. However, in state $s$:
\begin{itemize}
    \item $a$ considers it possible that we are in state $t$, where it would not be safe for her to announce $p$ since then $b$ would know that it was her.
    \item Thus it is not safe for $a$ to announce $p$ in s.
    \item $c$ knows that it is not safe for $a$ to announce $p$ in $s$, and thus it is not safe for $c$ either: if she announces $p$ in $s$, $b$ would know that it was her since $b$ knows that it is not safe for $a$ to announce $p$ in $s$.
    \item It is thus not safe for $b$ to announce $p$ in $s$ either
    \item It is not safe for $a$ or $b$ to announce $p$ in $t$ (only two agents)
    \item $p$ cannot be announced by anyone in $u$.
\end{itemize}
Thus, no safe announcements can be made in this model either.

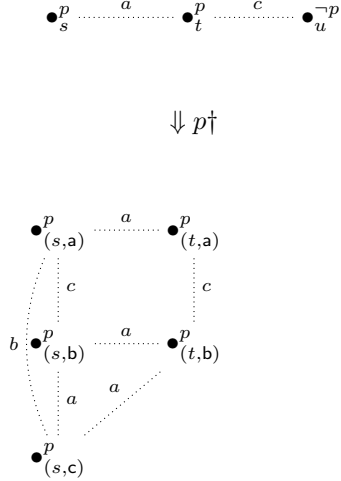
\begin{figure}
    \centering
\[\xymatrix{
\bullet^{p}_s\ar@{..}[r]^{a}&\bullet^{p}_t\ar@{..}[r]^{c}&\bullet^{\neg p}_u\\
&\Downarrow p\sa\\
\bullet^{p}_{(s,\sf a)}\ar@{..}@/_1pc/[dd]_{b}\ar@{..}[d]^{c}\ar@{..}[r]^{a}&\bullet^{p}_{(t,\sf a)}\ar@{..}[d]^{c}\\
\bullet^{p}_{(s,\sf b)}\ar@{..}[d]^{a}\ar@{..}[r]^{a}&\bullet^{p}_{(t,\sf b)}\\
\bullet^{p}_{(s,\sf c)}\ar@{..}[ur]^{a}
}\]
    \caption{Three-agent epistemic model, with update. Accessibility for agent $b$ is the identity relation.}
    \label{fig:intentions1}
\end{figure}

These examples show that it is not enough that three agents know $\phi$ to ensure that $\phi$ can be safely announced by any of them; to be safe those three agents not only need to know $\phi$ but also that $\phi$ is safe, i.e., that three agents know $\phi$ and that $\phi$ is safe \ldots and so on. This clearly has the flavour of \emph{common knowledge}. However, the existence of a group of three agents having common knowledge of $\phi$, \[M, s\models \bigvee_{a \neq b \neq c} C_{\{a,b,c\}}\phi\]
is sufficient but not necessarily for $\phi$ to be safe in $s$. A weaker condition would also be suffient: we don't need the three agents to know that \emph{the same} three agents safely can announce $\phi$. It might be, for example, that $a$ considers it possible that $a$, 
$b$ and $d$ safely can announce $\phi$.

In order to define that weaker condition, let us recall the fixpoint definition of common knowledge (see, e.g., \cite{Fagin:1995hc}). It is well known that $C_G\phi$ is the greatest fixpoint of \[E_G(\phi \wedge x)\]
w.r.t. a variable $x$ not occurring in $\phi$, or, more accurately, given an epistemic model $M$ the extension $\ext{C_G\phi}$ (the set of states where $C_G\phi$ is true) is the greatest fixpoint of the function \[f(S) = \ext{E_G(\phi \wedge x)}^M_{V[x = S]}\]
where $\ext{\psi}^M_{V[x = S]}$ is the extension of $\psi$ in the model $M$ where the valuation function has been changed so that the extension of the variable $x$ is $S$. Using a well known property of the fixed-point of monotonic functions, i.e., that the greatest fixed-point is equal to the union of all post-fixed points, we have that: \[\ext{C_G\phi} = \bigcup\left\{S : S \subseteq \ext{E_G(\phi \wedge x)}^M_{[x = S]}\right\}.\]
In terms of the modal $\mu$-calculus \cite{bradfield_modal_2007}, this can be written explicitly: \[C_G\phi \leftrightarrow \nu x.E_G(\phi \wedge x).\]

We now define a similar, weaker, notion of common knowledge in order to capture safety. Let us introduce a safety operator: $\safe \phi$ means that $\phi$ is safe in the current state. We let, where $N^3 = \{\{a,b,c\} : a,b,c \in N; a \neq b, a \neq c, b \neq c\}$ is the set of all groups of three different agents:
\[M,s \models \safe \phi \Leftrightarrow s \in \bigcup\left\{S : S \subseteq \ext{\bigvee_{G \in N^3} E_G(\phi \wedge x)}^M_{[x = S]}\right\}.\]
In other words, $\safe \phi$ is the greatest fixpoint of
$\bigvee_{G \in N^3} E_G(\phi \wedge x)$.
The following shows that this definition makes sense, w.r.t. safety. It is similar to the iterative definition of common knowledge.

\begin{lemma}\mbox{}
\label{lemma:iteration}
\[\begin{array}{ll}
    M,s \models \safe \phi \text{ iff}
        &M,s \models \phi \text{ and }\\
        &M,s \models \bigvee_{G_1 \in N^3} E_{G_1}\phi \text{ and}\\
        &M,s \models \bigvee_{G_1 \in N^3} E_{G_1} (\phi \wedge \bigvee_{G_2 \in N^3} E_{G_2}\phi)\text{ and}\\
        &M,s \models \bigvee_{G_1 \in N^3} E_{G_1} (\phi \wedge \bigvee_{G_2 \in N^3} E_{G_2}(\phi\wedge \bigvee_{G_3 \in N^3} E_{G_3}\phi)) \text{ and}\\
&    \cdots\\
        \end{array}\]
\end{lemma}

Lemma \ref{lemma:iteration} also hints at another point of similarity between common knowledge and our new operator that the reader may find helpful. Recall that $C_G\phi$ is the reflexive transitive closure of $\sim_a$ for all $a\in G$. Intuitively, this means that when interpreted on a model, $C_G\phi$ holds if and only if all paths labelled by any $a\in G$ only lead to $\phi$-states. For safety, the difference is that now at each step of a path we choose a triple of agents  such that \textit{all} their relations lead to a $\phi$-state. The triple of agents can be different at each step.

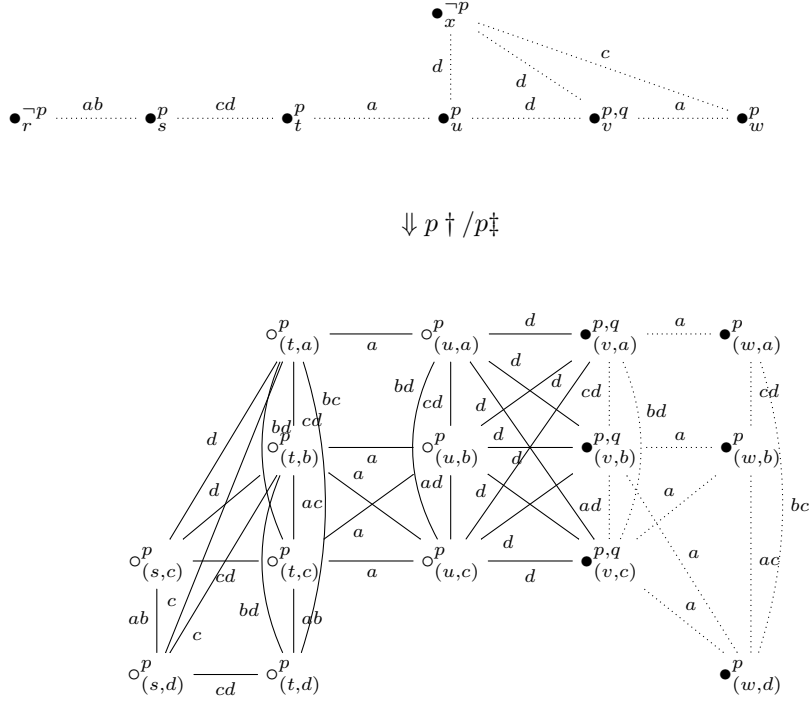
\begin{figure}[t!]
    \centering
\[\xymatrix{
&&&\bullet^{\neg p}_x\ar@{..}[dr]_{d}\\
\bullet^{\neg p}_r\ar@{..}[r]^{ab}&\bullet^{p}_s\ar@{..}[r]^{cd}&\bullet^{p}_t\ar@{..}[r]^{a}&\bullet^{p}_u\ar@{..}[r]^{d}\ar@{..}[u]^{d}&\bullet^{p,q}_v\ar@{..}[r]^{a}&\bullet^{p}_w\ar@{..}[ull]_{c}
\\
&&&\Downarrow p\sa/p\sai\\
&&\bbullet^{p}_{(t,a)}\ar@{-}[r]_{a}\ar@{-}[d]^(.75){cd}\ar@{-}@/_1pc/[dd]^(.4){bd}\ar@{-}@/^1pc/[ddd]^(.2){bc}&\bbullet^{p}_{(u,a)}\ar@{-}[r]^{d}\ar@{-}[rd]^(.35){d}\ar@{-}[rdd]_{d}\ar@{-}[d]_(.6){cd}\ar@{-}@/_1.2pc/[dd]_(.25){bd}&\bullet^{p,q}_{(v,a)}\ar@{..}[r]^{a}\ar@{..}[d]_{cd}\ar@{..}@/^1pc/[dd]^(.35){bd}&\bullet^{p}_{(w,a)}\ar@{..}[d]^{cd}\ar@{..}@/^1pc/[ddd]^{bc}\\
&&\bbullet^{p}_{(t,b)}\ar@{-}[r]_{a}\ar@{-}[rd]^(.35){a}\ar@{-}[d]^{ac}\ar@{-}@/_1pc/[dd]_(.7){bd}&\bbullet^{p}_{(u,b)}\ar@{-}[ru]^(.25){d}\ar@{-}[r]^(.3){d}\ar@{-}[rd]_(.25){d}\ar@{-}[d]_(.3){ad}&\bullet^{p,q}_{(v,b)}\ar@{..}[d]_{ad}\ar@{..}[r]^{a}\ar@{..}[ddr]^{a}&\bullet^{p}_{(w,b)}\ar@{..}[dd]^{ac}\\
&\bbullet^{p}_{(s,c)}\ar@{-}[d]_{ab}\ar@{-}[ruu]^{d}\ar@{-}[ru]^{d}\ar@{-}[r]_{cd}&\bbullet^{p}_{(t,c)}\ar@{-}[r]_{a}\ar@{-}[ru]_(.35){a}\ar@{-}[d]^{ab}&\bbullet^{p}_{(u,c)}\ar@{-}[ruu]^(.75){d}\ar@{-}[ru]_(.3){d}\ar@{-}[r]_{d}&\bullet^{p,q}_{(v,c)}\ar@{..}[dr]^{a}\ar@{..}[ur]^{a}\\
&\bbullet^{p}_{(s,d)}\ar@{-}[ruuu]^(.20){c}\ar@{-}[ruu]_(.20){c}\ar@{-}[r]_{cd}&\bbullet^{p}_{(t,d)}&&&\bullet^{p}_{(w,d)}
}\]
    \caption{Four-agent epistemic model $M$, with updated model $M^{p\sai}$ (black nodes and dotted edges) as well as $M^{p\sa}$ (black and white nodes, solid and dotted edges).}    \label{fig:intentions2}
\end{figure}

While $\safe\phi$ means that $\phi$ can safely be announced, $K_a\safe$ means that $\phi$ can safely be announced \emph{by $a$}. It is easy to see that:
\[\begin{array}{ll}
    M,s \models K_a\safe \phi \text{ iff}
        &M,s \models \phi \text{ and }\\
        &M,s \models \bigvee_{G_1 \in N^3, a \in G_1} E_{G_1}\phi \text{ and}\\
        &M,s \models \bigvee_{G_1 \in N^3, a \in G_1} E_{G_1} (\phi \wedge \bigvee_{G_2 \in N^3} E_{G_2}\phi)\text{ and}
        \cdots
        \end{array}\]

Finally we can introduce an operator for intentional anonymous announcements, $[\phi\sai]$. The semantics is defined as follows. The model update is defined by strengthening the pre-condition in the definition of the updated model from the previous section from $K_a\phi$ to $K_a \safe \phi$ -- $a$ knows that $\phi$ is safe.

\begin{definition}
    The update of epistemic model $M = (S,\sim,V)$ by the safe pseudo-anonymous announcement of $\phi$ is the epistemic model $M^{\phi\sai} = (S',\sim',V')$ where:
    \begin{itemize}
        \item $S' = \{(s,a) : s \in S$, $M,s \models K_a \safe\phi$\}
        \item $(s,a) \sim'_c (t,b)$ iff $s \sim_c b$ and $a = c$ iff $b = c$ 
        \item $V'(s,a) = V(s)$
    \end{itemize}
\end{definition}
We then let:
\[
M,s \models [\phi\sai]\psi \Leftrightarrow \forall a \in N, \left(M,s \models K_a \safe\phi \Rightarrow M^{\phi\sai},(s,a) \models \psi\right).\]

Thus, we have introduced two new operators: $\safe$ and $[\phi\sai]$. Do we want both in the formal language? Let's consider 
languages for all three combinations, which turns out to be useful later\footnote{\sailang\ is well defined without $\safe$ in the syntax, even though $\safe$ is used in the semantic condition for $[\phi\sai]$. That semantic condition could of course be written without mentioning $\safe$, replacing it with its meaning.}. 
\[
\begin{array}{ll}
\sailang& \phi ::= p \mid \neg \phi \mid \phi \wedge \phi \mid K_i \phi \mid [\phi\sai] \phi\\
\safelang& \phi ::= p \mid \neg \phi \mid \phi \wedge \phi \mid K_i \phi \mid \safe \phi\\
\saifelang& \phi ::= p \mid \neg \phi \mid \phi \wedge \phi \mid K_i \phi \mid [\phi\sai] \phi \mid \safe \phi
\end{array}\]

Let's look at some examples. In the model in Figure \ref{fig:intentions2}:
\begin{itemize}
    \item $M,s \models \neg\safe p$. $p$ can't be safely announced in $s$: only two agents know $p$.
    \item $M,t \models \neg\safe p$. While all four agents know $p$ in $t$, they don't all know that three agents know $p$: $c$ and $d$ consider it possible that only two agents know $p$. In other words, even though they know $p$ they don't know that it can be safely announced, because they consider it possible that we are in $s$ in which if can't be.
    \item $M,u \models \neg\safe p$. Three agents, $a$, $b$ and $c$, know $p$. They also know that at least three agents know $p$. But they still don't know that $p$ is safe: they don't know that at least three agents know that at least three agents know $p$.
    \item $M,v \models \safe p$ and $M,w \models \safe p$. $p$ can be safely announced in $v$ and in $w$. In fact, $\{w,v\}$ is the greatest fixed-point of $\bigvee_{G \in N^3} E_G(\phi \wedge x)$. 
    \item $M,v \models K_a\safe p$. $p$ can be safely announced by $a$ in $v$.
    \item Even though $a$ knows \emph{that} $p$ is safe in $v$, she doesn't know \emph{why}. She considers it possible that we indeed are in $v$ in which it would be safe for $a$, $b$ and $c$ to announce $p$, or in $w$ in which it would be safe for $a$, $b$ and $d$ to announce $p$, but she doesn't know which of the two states we are in. But she knows that no matter which state we are in, it is safe to announce $p$ -- but for different reasons. We can say that she knows \emph{de dicto} that $p$ is safe, but not \emph{de re}. 
    \item $M,v \models [p\sai]K_cK_d q$. After an intentionally anonymous announcement of $p$ in $w$, $c$ knows that $d$ has learned that $q$ is true.
    \item Safe anonymous announcements are \emph{more informative} than general pseudo-anonymous announcements. 
    \item In every state in the updated model, \emph{only the announcer knows who the announcer was}.
\end{itemize}

The latter point is no coincidence. The following shows that safe announcements are indeed safe:

\begin{lemma}
    \label{lemma:ck}
    In any update by a safe pseudo-anonymous announcement, it is common knowledge that no-one except the announcer knows who the announcer is (more technically: in any state $(s,a)$, for any agent $i \neq a$ there is a state $(s',a')$ such that $(s,a) \sim_i (s',a')$ and $a \neq a'$).
\end{lemma}
\begin{proof}
    Let $M$ be an epistemic model. Let $(s,a)$ be a state in $M^{\phi\sai}$ and let $d \neq a$. By definition, $M,s \models K_a\safe\phi$, so there are $b \neq c$, $b \neq a$ and $c \neq a$, such that $M,s \models K_b \safe\phi \wedge K_c \safe\phi$. That means that $(s,b)$ and $(s,c)$ are states in the updated model. We must have that either $b \neq d$ or $c \neq d$ (or both); assume the former. Since all of $a$, $b$ and $d$ are different, we have that $(s,a) \sim'_d (s,b)$.
\end{proof}

Thus, $K_a\safe\phi$ is a \emph{sufficient} condition for a safe anonymous announcement: after it is announced, no one will know who the announcer is. It is also \emph{necessary}, in the sense that $a$ must know that at least two other agents know $\phi$ (if there is only one ``the other" agent will know who made the announcement), and must know that those two agents know that at least two other agents know, and so on -- $K_a\safe\phi$. That justifies the definition of the model update: $K_a\safe\phi$ is the information that is revealed.

The reader might have observed that the updated model $M^{p\sai}$ in Figure \ref{fig:intentions2} is bisimilar to the submodel of $M$ it is projected on, i.e., the submodel consisting of states $v$ and $w$ -- exactly the submodel where $K_a\safe p$ holds for some $a$ in other words. That is no coincidence: it is in fact always the case\footnote{Unlike for $M^{\phi\sa}$, see Figure \ref{fig:update}.}. Indeed, we in fact have the following.

\begin{lemma}
\label{lemma:anpal-pal}
    For any $\phi$, $\psi$, $M$, and $s$, 
    \[M,s\models [\phi\sai]\psi \text{ iff } M,s\models \safe\phi \Rightarrow M^{\safe\phi!},s \models \psi
    \]
\end{lemma}
\begin{proof}
    We first show that $M^{\phi\sai}$ and $M^{\safe\phi!}$ are bisimilar. Let $Z$ be such that $(s,a)Zs'$ iff $s = s'$
     and let $(s,a)Zs$. \emph{Atoms} is straightforward. \emph{Forth} is as well: if $(s,a) \sim_c (t,b)$, then $s \sim_c t$ by definition.  For \emph{Back}, let $s\sim_c t$. First, consider the case that $c \neq a$. Since $t$ is in $M^{\safe\phi!}$, $M,t \models \safe\phi$. Thus, there are $b,c,d \in N$ (all different) such that $M,t \models K_b \safe\phi \wedge K_c \safe \phi \wedge K_d \safe\phi$. At least one of $b,c,d$ must be different from both $a$ and $c$. Say, $b$. Thus $(s,a) \sim_c (t,b)$. Second, consider the case that $c = a$. Since $M,s \models K_c\safe$ and $s \sim_c t$, we also have that $M,t \models K_c\safe$. Thus $(t,c)$ is in $M^{\phi\sai}$, and we have that $(s,a) \sim_c (t,c)$. 

    Note that satisfaction of $\safe$ is invariant under bisimulation (this follows, e.g., from Lemma \ref{lemma:iteration}).
    Then we have that $M,s \models [\phi\sai]\psi$ iff for all $a$ $M,s \models K_a\safe\phi$ implies $M^{\phi\sai},(s,a) \models \psi$ iff for all $a$ $M,s \models K_a\safe\phi$ implies $M^{\safe\phi!},s \models \psi$. We argue that this is equivalent to $M,s \models \safe\phi$ implies $M^{\safe\phi!},s \models \psi$. For the implication from left to right, if $M,s \models \safe \phi$ then $M,s\models K_a\safe\phi$ for some $a$ by the fixed-point property of $\safe$ (actually that holds for three different $a \in N$). From $M,s \models K_a\safe\phi \Rightarrow M^{\safe\phi!},s \models \psi$ for all $a$, we get that $M^{\safe\phi!},s \models \psi$. For the other direction, let $a \in N$. If $M,s \models K_a\safe\phi$, then $M,s\models \safe\phi$ by reflexivity, so $M^{\safe\phi!},s \models \psi$. 
\end{proof}

As an immediate corollary of Lemma \ref{lemma:anpal-pal} we have that  if we allow public announcement operators in the language, then the following holds: \[\models [\phi\sai]\psi \leftrightarrow [\safe\phi!]\psi.
\]

So, safe anonymous announcements are exactly public announcements of safety. Thus, we get a ``simpler", alternative, logically equivalent semantics. The existence of the original semantics and the fact that it corresponds exactly to this alternative semantics is of course still crucial: it is what ensures that safety is safe (Lemma \ref{lemma:ck})\footnote{We used invariance under standard bisimulation as a technical tool to establish the correspondence. Of course, when we say that ``only the announcer knows who the announcer is" and so on, we implicitly mean agents can potentially discern between states $(s,a)$ and $(s,b)$, even though, as shown above, they are bisimilar. An extended notion of bisimulation could be introduced to take into account this assumption, but it is in any case not picked up by the logical language, and here we only need standard bisimulation as a technical tool. It is of course also possible to think about ways this distinction \emph{could} be picked up by the language, e.g., by introducing special atoms, but that is complicated by, e.g, the possibility of iterated announcements leading to states like $((s,a),b)$ and so on.}.

It also means that PAL extended with $\safe$ can express safe anonymous announcements. However, as we shall see when we now move on to comparing the expressive power of the three languages introduced above, even less is needed.

\subsection{Expressive power}

Consider the three languages $\saifelang, \safelang$ and $\sailang$. The first is similar to AMLC -- it is epistemic logic extended with (restricted) action model operators and the common-knowledge-like operator $\safe$. In the same sense, $\safelang$ is similar to ELC. $\sailang$ is in-between, it would correspond to AML with common knowledge allowed in pre-conditions. AMLC is strictly more expressive than ELC \cite[Chapter 8]{vDvdHK2007}; in particular expressions of the form $[M,s]C_G \phi$ cannot 
always be reduced to an ELC formula. Nevertheless, we shall now show that, perhaps contrary to intuition,  \emph{all three languages $\saifelang, \safelang$ and $\sailang$ are actually all equally expressive}.

Let us start with the ``intermediate" language $\sailang$. It is easy to see that $\phi\sai$ can be expressed by action models like in Section \ref{sec:actionmodels}, except that the preconditions now are $K_a\safe\phi$ instead of $K_a\phi$. We saw that $\salang$ is reducible to epistemic logic, in the sense that for every formula in $\salang$ there is an equivalent formula in epistemic logic, because we can translate the $\salang$ formula to an AML formula which we know we can translate to an epistemic logic formula by using reduction axioms for AML. The same argument cannot be used directly for $\sailang$, because the preconditions, that the reduction axioms use, include the $\safe$ operator, which is not a part of the language. And if it was, as in the $\saifelang$ language, we would need a reduction axiom for $\safe$ as well.  The corresponding result does, however, still hold.

\begin{lemma}
    \label{lemma:exp1}
    For any $\phi \in \sailang$, there is a $\phi' \in \safelang$ such that for any $M,s$, it holds that $M,s \models \phi$ iff $M,s \models \phi'$.
\end{lemma}
\begin{proof}
    Consider the language $AML\safe^-$, which is the language of AML extended with $\safe$, but with the restriction that $\safe$ is not allowed in the scope of a $[M,s]$ modality. It is easy to see that the reduction axioms for AML still are valid for this extended language. Take the formula $\phi$, let $\psi$ be the corresponding AML formula.  Now we recursively reduce the number of occurrences of $[\cdot]$ modalities in $\psi$ by starting with an \emph{outermost} one, i.e., one that is not in the scope of any other. Use the corresponding reduction axiom, and repeat until all $[\cdot]$ modalities are gone. Every time we use a reduction axiom we might introduce a new subformula $pre(s) = K_a\safe\psi'$, but never in the scope of a $[\cdot]$ modality since we started with the outermost modality. Thus, every time we use a reduction axiom the result is in $AML\safe^-$, and when we are done the result is in pure epistemic logic extended with $\safe$ -- which is $\safelang$.
\end{proof}

Thus, $\sailang$ can be ``reduced" to $\safelang$ (even though the latter is not a sub-language). 

The second ``surprise" is that the $\safe$ operator can actually be expressed by $[\phi\sai]$.

\begin{lemma}
    \label{lemma:exp2}
    For any $M,s$ and any $\phi$ (in any of the languages), $M,s \models \safe \phi$ iff $M,s \models \neg [\phi\sai]\bot$.
\end{lemma}
\begin{proof}
    $M,s \models \neg [\phi\sai]\bot$ iff there is an $a \in N$ such that $M,s \models K_a\safe\phi$ and $M^{\phi\sai}, (s,a) \not\models \bot$ iff there is an $a \in N$ such that $M,s \models K_a\safe\phi$ iff (by reflexivity in one direction and the fixed-point definition of $\safe$ in the other) $M,s \models \safe \phi$.
\end{proof}

We thus get the mentioned result.
\begin{corollary}
    $\saifelang\approx \sailang \approx \safelang$.
\end{corollary}
\begin{proof}
    From Lemma \ref{lemma:exp2} we get that $\sailang \approx \saifelang$. From Lemma \ref{lemma:exp1} we get that $\sailang \preccurlyeq \safelang$, and from Lemma \ref{lemma:exp2} again that $\safelang \preccurlyeq \sailang$.
\end{proof}

Going back to the comparison with AMLC above and the mentioned non-reduction of $[M,s]C_G\psi$, the difference is that in $[\phi\sai]\safe\psi$ we can replace $\safe\psi$ with $\neg [\psi\sai]\bot$ and then reduce $[\phi\sai]\neg [\psi\sai]\bot$. This reduction might re-introduce $\safe$-operators, but outside the scope of any announcement modalities.

Thus, again the dynamic anonymous announcement operators can actually be expressed in a purely static language: $\safelang$. We thus move on to axiomatising safety.

\subsection{Axiomatisation of safety}

Observe that $\safe$ does not distribute over implication: $\not\models\safe(\phi \rightarrow \psi) \rightarrow (\safe \phi \rightarrow \safe \psi)$.
In particular, $\bigvee_{G \in N^3}(\phi \rightarrow \psi)$ and $\bigvee_{G \in N^3}\phi$ does not imply $\bigvee_{G \in N^3}\psi$ -- it might not be the same $G$ in the first two cases. For similar reasons the conjunctive closure axiom, $(\safe\phi \wedge \safe\psi)\rightarrow \safe(\phi \wedge \psi)$ does not hold. This is similar to \emph{somebody knows} \cite{sk}, but unlike most other group knowledge operators which are normal modalities. The other direction of conjunctive closure, \emph{monotonicity}, $\safe(\phi \wedge \psi) \rightarrow (\safe\phi \wedge \safe\psi)$, does hold (again, similarly to somebody knows). 

The axiomatic system $\safesys$ for $\safelang$ is shown in Table \ref{tab:syssafe}. The first two parts is a standard axiomatisation of propositional logic and the individual knowledge operators. The Monotonicity rule combines the monotonicity axiom and the replacement of equivalents rule standard in weak modal logics. We note that necessitation for $\safe$, from $\phi$ derive $\safe\phi$, follows (we have $\safe\top$ from the Induction rule). The Mix axiom (or \emph{fixed-point axiom}) says that $\safe \phi$ is indeed a fixed-point of $\bigvee_{G \in N^3} E_G(\phi \wedge x)$. Finally, the Induction rule give us a way to derive $\safe\phi$. Mix and the Induction rule can be seen as adaptions of similar axioms/rules for common knowledge, see, e.g., \cite{Fagin:1995hc,vDvdHK2007}.

We will use the following shorthand:
\[\safe_n = \phi \wedge \bigvee_{G_1 \in N^3}E_{G_1}(\phi \wedge \bigvee_{G_2 \in N^3}E_{G_2}(\phi \wedge \bigvee_{G_3 \in N^3}E_{G_3}(\phi \wedge \cdots \wedge \bigvee_{G_n \in N^3}E_{G_n}\phi))).\]
Thus, $\safe_0 = \phi$, $\safe_1 = \phi \wedge \bigvee_{G_1 \in N^3}E_{G_1}\phi$, and so on, and
\begin{equation}
    \label{eq:sem}
    M,s \models \safe \phi \text { iff } M,s \models \safe_n \phi \text { for all $n$}.
    \end{equation}

\begin{table}
    \centering
    \begin{tabular}{ll}\toprule
        all instances of propositional tautologies& Prop\\
        From $\phi \rightarrow \psi$ and $\phi$, derive $\psi$& Modus ponens\\
        \midrule
        $K_a(\phi \rightarrow \psi) \rightarrow (K_a\phi \rightarrow K_a \psi)$& \text{Distribution}\\
        $K_a \phi\rightarrow \phi$& \text{Truth}\\
        $\neg K_a\phi \rightarrow K_a\neg K_a\phi$& \text{Negative introspection}\\
        From $\phi$, derive $K_a\phi$& Necessitation\\ 
        \midrule
        $\safe \phi \rightarrow \bigvee_{G \in N^3} E_G(\phi \wedge \safe \phi)$& \text{Mix}\\
        From $\phi \rightarrow \bigvee_{G \in N^3}E_G\phi$, derive $\phi \rightarrow \safe \phi$& Induction\\
        From $\phi \rightarrow \psi$, derive $\safe \phi \rightarrow \safe \psi$& Monotonicity\\
    \bottomrule
    \end{tabular}
    \caption{The proof system \safesys.}
    \label{tab:syssafe}
\end{table}

\subsubsection{Soundness}

\begin{lemma}[Mix Axiom]
    \label{lemma:mix}
    $\models \safe \phi \rightarrow \bigvee_{G \in N^3} E_G(\phi \wedge \safe\phi)$
\end{lemma}
\begin{proof}
    By definition, $\ext{\safe\phi}$ is a fixed-point of the function $f_{\bigvee_{G \in N^3}(\phi \wedge x)}(A)$ that maps a set of states $A$ to the set of states where $\bigvee_{G \in N^3}(\phi \wedge x)$ is true when $x$ is true exactly in $A$. It follows immediately that $M,s \models \safe\phi$ iff $M,s \models \bigvee_{G \in N^3} E_G (\phi \wedge \safe \phi)$.
\end{proof}

\begin{lemma}[Induction Rule]
    \label{lemma:indrule}
    If $\models \phi \rightarrow \bigvee_{G \in N^3}E_G\phi$, then $\models \phi \rightarrow \safe \phi$.
\end{lemma}
\begin{proof}
    Let $\models \phi \rightarrow \bigvee_{G \in N^3}E_G\phi$ and let $M,s \models \phi$. We show that $M,s \models \safe_n$ for all $n$, by induction on $n$. The base case, $n = 0$, is immediate. Assume that $M,s \models \safe_n$, i.e., that 
    \[M,s \models \phi \wedge \bigvee_{G_1 \in N^3}E_{G_1}(\phi \wedge \bigvee_{G_2 \in N^3}E_{G_2}(\phi \wedge \bigvee_{G_3 \in N^3}E_{G_3}(\phi \wedge \cdots \wedge \bigvee_{G_n \in N^3}E_{G_n}\phi))).\]
    It follows from validity of $\phi \rightarrow \bigvee_{G \in N^3}E_G\phi$ that also
        \[M,s \models \phi \wedge \bigvee_{G_1 \in N^3}E_{G_1}(\phi \wedge 
        \cdots \wedge \bigvee_{G_n \in N^3}E_{G_n}(\phi \wedge \bigvee_{G_{n+1}\in N^3}E_{G_{n+1}}\phi)),\]
        i.e., $M,s \models \safe_{n+1}\phi$.
\end{proof}

We thus 
get the following.
\begin{proposition}
    \safesys is sound.
\end{proposition}
\begin{proof}
    Validity (preservation) for the two first parts of the system is standard. Monotonicity is straightforward. Validity of the Mix axiom and validity preservation for the induction rule were shown in Lemmas \ref{lemma:mix} and \ref{lemma:indrule}, respectively. 
    \end{proof}

    \subsubsection{Completeness}

We now prove that $\safesys$ is also complete. Similar to epistemic logic with common knowledge, it is easy to see that the logic is not \emph{compact}. For example, \[\{\safe_n\phi : n \geq 0\} \cup \{\neg \safe \phi\}\]
is not satisfiable, but any finite subset of it is. Thus, we have to settle for \emph{weak} rather than \emph{strong} completeness. To that end, we will adapt the standard technique for common knowledge of defining a finite canonical model based on the finite syntactic closure of some formula. There is a complication: unlike $C_G$, $\safe$ is not a normal modality. This shows up in the axiomatisation: the K axiom is replaced with the (weaker) monotonicity rule. The consequence, aside from the fact that we cannot rely on distribution over implication like in proofs for common knowledge, is that there is no normal relational semantics for $\safe$ (like ``$\phi$ is true on all states reachable by a $G$-path" for $C_G\phi$). We now define an alternative definition of the semantics and show that it is equivalent, that we will use in the completeness proof.

Given a model $M$, a \emph{group assignment function} for $M$ is a function $f:W \rightarrow \pow(N)$ such that $f(s) \in N^3$ or $f(s) = \emptyset$ for all $s$, i.e., assigning a group of three agents to some of the states, such that for all $s,t \in W$ and $i \in f(s)$, if $s \sim_i t$ then $f(t) \neq \emptyset$. The idea is that $f$ works as an assignment of groups of three agents to certain states, such that if we follow the accessibility relations for those agents, $f$ again assigns a group of three agents to those states, and so on. 

An \emph{$f$-consistent path} in $M$ is a finite sequence of states $s_0s_1s_2\cdots s_m$ such that for all $0 \leq i < m$, $s_i \sim_a s_{i+1}$ for some $a \in f(s_i)$. We say that $\phi$ is true on a path if it is true in every state on that path.

\begin{lemma}
    \label{lemma:f}
    $M,s \models \safe \phi$ iff there is a group assignment function $f$ such that $f(s) \neq \emptyset$ and $\phi$ is true on every $f$-consistent path starting in $s$.
\end{lemma}
\begin{proof}
    \textit{Left-to-right:} let $M,s \models \safe \phi$. Define $f$ as follows: for any state $t$, $f(t) = G$ for some $G$ such that $M,t \models E_G\safe \phi$ if such a $G$ exists, $f(t) = \emptyset$ if not. If there are several such $G$ chose any of them. We must show that $f$ is indeed a group assignment function. Let $i \in f(s')$ and $s' \sim_i t$; we must show that $f(t) \neq \emptyset$.  Since $i \in f(s')$, there is a $G\subseteq N$ such that $M,s' \models E_G \safe \phi$ and $i \in G$. Thus $M,t \models \safe\phi$. By 
    the Mix axiom, $M,t \models E_{G'}\safe \phi$ for some $G'$, so $f(t) \neq \emptyset$. Finally, $f(s) \neq \emptyset$ from $M,s \models \safe\phi$ and validity of Mix.
    
    We proceed by induction on the length of any $f$-consistent path starting in $s$. The base case is immediate. Consider an $f$-consistent path $s_0\cdots s_ns_{n+1}$ where $s_0 = s$. By the induction hypothesis $s_0, \cdots, s_n$ are $\phi$-states. Since $s_n \sim_a s_{n+1}$ for some $a \in f(s_n)$ and $M,s_n \models E_{f(s_n)} \safe \phi$, also $M,s_{n+1} \models \phi$.

    \textit{Right-to-left:} we reason by contraposition. Assume that $M,s \not\models \safe \phi$, i.e., that $M,s \not \models \safe_n \phi$ for some $n$. Let $f$ be a group assignment function such that $f(s) \neq \emptyset$. We must show that there is a $f$-consistent path starting in $s$ where $\phi$ is not true. We have that:
    \[M,s \models \neg\phi \vee \bigwedge_{G_1 \in N^3}\bigvee_{i_1 \in G_1}\hat{K}_{i_1}(\neg \phi \vee \bigwedge_{G_2 \in N^3}\bigvee_{i_2 \in G_2}\hat{K}_{i_2}(\neg \phi \vee 
    \cdots \vee \bigwedge_{G_n \in N^3}\bigvee_{i_n \in G_n}\hat{K}_{i_n}\neg\phi)).\]
    Thus, there is a $k \leq n$ such that
    \[M,s \models \bigwedge_{G_1 \in N^3}\bigvee_{i_1 \in G_1}\hat{K}_{i_1} \bigwedge_{G_2 \in N^3}\bigvee_{i_2 \in G_2}\hat{K}_{i_2}\bigwedge_{G_3 \in N^3}\bigvee_{i_3 \in G_3}\hat{K}_{i_3} \cdots \vee \bigwedge_{G_k \in N^3}\bigvee_{i_k \in G_n}\hat{K}_{i_k}\neg\phi.\]
    Thus there is a path $s_0s_1\cdots s_{k+1}$ where:
    \[\begin{array}{lll}
        s_0 = s\\
        s_0 \sim_{i_1} s_1 & \text{ some } i_1 \in f(s_0)\\
        s_1 \sim_{i_2} s_2 & \text{ some } i_2 \in f(s_1)\\
        \cdots\\        
        s_k \sim_{i_{k+1}} s_{k+1} & \text{ some } i_{k+1} \in f(s_k)
        \end{array}
        \]
    such that $M,s_{k+1} \models \neg \phi$. This is an $f$-consistent path.
\end{proof}

We now proceed with defining the finite canonical model, and proving a truth lemma. We adapt the standard proof for common knowledge (see \cite[Chapter 7]{vDvdHK2007}\footnote{In addition to the complications mentioned above, there are other differences, including that the correspondent to the \emph{induction axiom} for common knowledge, $\safe(\phi \rightarrow \bigvee_{G \in N^3} \phi) \rightarrow (\phi \rightarrow \safe\phi)$, does not hold. We use a variant of the induction rule for common knowledge, used in, e.g., \cite{Fagin:1995hc}}).

\begin{definition}
    The closure $\cl(\gamma)$ of a formula $\gamma$ is the smallest set that contains all subformulas of $\gamma$, is closed under single negations, and that contains $E_G\safe \phi$ for all $G \subseteq N$ whenever it contains $\safe \phi$.
\end{definition}
The closure of any formula is finite. As usual we say that a a set of formulas $\Gamma$ is maximal consistent in $\cl(\gamma)$ for some given $\gamma$ iff $\Gamma \subseteq \cl(\gamma)$, $\Gamma$ is consistent, and there is no consistent $\Gamma' \subseteq \cl(\gamma)$ such that $\Gamma \subset \Gamma'$. When $\Delta$ is a set of formulas we write $\underline{\Delta}$ for $\bigwedge_{\delta \in \Delta}\delta$.

The canonical model is defined as follows.

\begin{definition}
    The canonical model for some formula $\gamma$ is $M^\gamma = (S^\gamma, \sim^\gamma, V^\gamma)$ where:
    \begin{itemize}
        \item $S^\gamma$ is the set of all sets of formulas maximal consistent in $\cl(\gamma)$,
        \item $\Gamma \sim^\gamma \Delta$ iff $\{K_i \phi : K_i \phi \in \Gamma\} = \{K_i \phi : K_i\phi \in \Delta\}$, and
        \item $V^\gamma(p) = \{\Gamma \in S^\gamma : p \in \Gamma\}$.
    \end{itemize}
\end{definition}

This definition of the canonical model is identical to the one used for common knowledge \cite{vDvdHK2007}, except that the closure is wider (it contains $K_i\safe \phi$ for every $i$, instead of $K_iC_G\phi$ for every $i \in G$). Many of the properties of that model, like deductive closure of maximal consistent sets in $cl(\gamma)$, carry over, and we only need to focus on the following property for the $\safe$ modality.

A $\phi$-path in the canonical model is a path $\Gamma_0\Gamma_1\cdots$ where $\phi \in \Gamma_i$ for every $i$.

\begin{lemma}
    \label{lemma:safe}
    For any formula $\gamma$ and $\Gamma \in S^\gamma$, if $\safe\phi \in \cl(\gamma)$ then $\safe \phi \in \Gamma$ iff there is a group assignment function $f$ such that $f(\Gamma) \neq \emptyset$ and every $f$-consistent path from $\Gamma$ is a $\phi$-path.
\end{lemma}
\begin{proof}
    \textit{Left-to-right}: let $\safe\phi \in \Gamma$. We define $f$ as follows: for any $\Delta \in S^\gamma$, $f(\Delta) = G$ if $E_G\safe\phi \in \Delta$ for some $G \subseteq N$ and $f(\Delta) = \emptyset$ otherwise. If there are several such $G$, chose one of them. Consider a $\Delta$ such that $f(\Delta) \neq \emptyset$ and $\Delta \sim_i^\gamma \Delta'$ for some $i \in f(\Delta)$. We must show that $f(\Delta') \neq \emptyset$. Since $K_i\safe\phi \in \Delta$, $K_i \safe\phi \in \Delta'$ by definition of $\sim_i^\gamma$, and thus $\safe\phi \in \Delta'$ by the Truth axiom, and $E_{G'} \safe\phi \in \Delta'$ for some $G'$ by Mix and the fact that $E_{G'}\safe\phi\in \cl(\gamma)$. Finally, $f(\Gamma) \neq \emptyset$ also follows from Mix and the definition of the closure.
    
    We must show that every $f$-consistent path from $\Gamma$ is a $\phi$-path. The proof is by induction on the length of the path. For the base case, we have that $\phi \in \Gamma$ from the Mix and Truth axioms. For the induction step, consider an $f$-consistent path $\Gamma_0\cdots\Gamma_n\Gamma_{n+1}$ where $\phi \in \Gamma_j$ for all $j \leq n$. From $f$-consistency we have that $\Gamma_n \sim^\gamma_i \Gamma_{n+1}$ for some 
    $i \in G$ for some $G$ such that $E_G\safe\phi \in \Gamma_n$. It follows from the definition of $\sim_i^\gamma$ that $K_i\safe\phi \in \Gamma_{n+1}$, and thus that $\safe\phi \in \Gamma_{n+1}$ from the Truth axiom, and thus that $\phi \in \Gamma_{n+1}$ from the Mix axiom and Truth axiom again.

    \textit{Right-to-left}: let $f$ be such that $f(\Gamma) \neq \emptyset$ and every $f$-consistent path is a $\phi$-path. Let $S_{f,\phi}$ be the set of all sets $\Delta$ maximal consistent in $\cl(\gamma)$ such that $f(\Delta) \neq \emptyset$ and every $f$-consistent path from $\Delta$ is a $\phi$-path. Let
    $\chi = \bigvee_{\Delta \in S_{f,\phi}} \underline{\Delta}$.
    We first show that
    \begin{equation}
        \label{eq:c}
        \vdash \chi \rightarrow \bigvee_{G \in N^3}E_G\chi
    \end{equation}
    Assume, towards a contradiction, that $\chi \wedge \neg \bigvee_{G \in N^3}E_G\chi$ is consistent. Then $\underline{\Delta} \wedge \neg \bigvee_{G \in N^3}E_G\chi$ is consistent for some $\Delta \in S_{f,\phi}$. Then, for any $G \in N^3$ there must be a $i_G \in G$ such that $\underline{\Delta} \wedge \hat{K}_{i_G}\neg\chi$ is consistent.  $\underline{\Delta} \wedge \hat{K}_{i_G}\vee_{\Theta \in S^\gamma \setminus S_{f,\phi}} \underline{\Theta}$ is consistent for each $i_G$, and by modal reasoning for individual knowledge, $\underline{\Delta} \wedge \bigvee_{\Theta \in S^\gamma \setminus S_{f,\phi}} \hat{K}_{i_G}\underline{\Theta}$ is consistent for each $i_G$. Thus, for every $G \in N^3$ there is an $i_G \in G$ and a $\Theta_{i_G} \in S^\gamma \setminus S_{f,\phi}$ such that $\underline{\Delta} \wedge \hat{K}_{i_G} \underline{\Theta_{i_G}}$ is consistent. In particular, since $\Delta \in S_{f,\phi}$ and thus $f(\Delta) \neq \emptyset$, there is an $i_G \in G = f(\Delta)$ and a $\Theta_{i_G} \in S^\gamma \setminus S_{f,\phi}$ such that $\underline{\Delta} \wedge \hat{K}_{i_G} \underline{\Theta_{i_G}}$ is consistent. It follows (by a standard property of the canonical model \cite[Item 4 of Lemma 7.14]{vDvdHK2007}) that $\Delta \sim^\gamma_{i_G} \Theta_{i_G}$. Since $i_G \in f(\Delta)$, $f(\Theta_{i_G}) \neq \emptyset$ and since $\Theta_{i_G}$ is not in $S_{f,\phi}$ that means that there is a an $f$-consistent path from $\Theta_{i_G}$ which is not a $\phi$-path. It follows that there is an $f$-consistent path from $\Delta$ that is not a $\phi$-path. But $\Delta \in S_{f,\phi}$, which leads to a contradiction. 
    Thus, we have shown (\ref{eq:c}).

    From (\ref{eq:c}) and the induction rule, we get that
    $\vdash \chi \rightarrow \safe \chi$.
    Since $\Gamma$ is one of the disjuncts in $\chi$ we have that $\vdash \underline{\Gamma} \rightarrow \chi$. Thus, $\vdash \underline{\Gamma} \rightarrow \safe \chi$, and $\safe\chi \in \Gamma$.  Since $\phi \in \bigcap_{\Delta \in S_{f,\phi}} \Delta$, we also have that $\vdash \chi \rightarrow \phi$. By Monotonicity, $\vdash \safe \chi \rightarrow \safe \phi$, and since $\safe\phi \in \cl(\gamma)$, we also have that $\safe\phi \in \Gamma$.
    \end{proof}

\begin{lemma}[Truth]
    For any formulas $\gamma$ and $\phi$ and any $\Gamma \in S^\gamma$,
    $\phi \in \Gamma \text{ iff } M^\gamma,\Gamma \models \phi$.
\end{lemma}
\begin{proof}
    The proof is by induction on the structure of $\phi$. We only show the induction step for $\phi = \safe \psi$, the other cases are straightforward and/or standard.

    We have that $M^\gamma, \Gamma \models \safe\psi$ iff (by Lemma \ref{lemma:f}) there is an $f$ such that $f(\Gamma) \neq \emptyset$ and $\psi$ is true on every $f$-consistent path starting in $\Gamma$ iff (by the induction hypothesis) there is an $f$ such that $f(\Gamma) \neq \emptyset$ and every $f$-consistent path starting in $\Gamma$ is a $\psi$-path iff (by Lemma \ref{lemma:safe}) $\safe \psi \in \Gamma$. 
  \end{proof}
  
We immediately get the following.
\begin{theorem}[Completeness]
    $\safesys$ is complete.
\end{theorem}
\begin{proof}
    If $\neg\phi$ is consistent it can be extended to a set $\Gamma$ that is maximal consistent in $\cl(\neg\phi)$. From the truth lemma we get that $M^{\neg\phi},\Gamma \models \neg \phi$.
\end{proof}

\section{Discussion}
\label{sec:discussion}

Pseudo-anonymous public 
announcements $\phi\sa$ are made by an agent inside the system, but not explicitly ``signed" by that agent.
In a sense, such announcements are in-between public announcements $\phi!$ and $K_a\phi!$. 
These announcements are pseudo-anonymous, because, depending on the background knowledge of the 
agents in a system, the identity of the announcer might still be revealed. We showed that like normal public announcements, these operators can actually be expressed in (pure) epistemic logic. 

Similarly to the \textit{Russian Cards Problem} \cite{van2003russian} (see below), if we make assumptions about the \emph{intentions} of the announcer, we \emph{learn more}. Intentional anonymous announcements $\phi\sai$ capture exactly the announcements that are guaranteed to ensure anonymity. It all boils down to \emph{safety}: 
intentional anonymous announcements $\phi\sai$ are exactly public announcements of safety $\safe\phi!$. 
Furthermore, 
the $\phi\sai$ operators can be expressed in epistemic logic extended with only the safety operator $\safe$. We gave a complete axiomatisation of the latter.

There are still many open problems. While the expressivity results ``reduce" the languages with $\phi\sa$ and $\phi\sai$ to equivalent languages we have axiomatisations for, they don't directly give us ``native" axiomatisations of those languages themselves. For the case of pseudo-anonymous announcements, we don't get reduction axioms directly from AML, since the language is less expressive. We note that 
we would get reduction axioms (and thus completeness) immediately if we used a more fine-grained variant
\[[\phi\sa_a]\psi\]
meaning ``after $\phi$ is pseudo-anonymously announced by $a$ (``$a$ writes $\phi$ on the blackboard''), $\psi$ is true" -- 
corresponding to a single action model point (instead of a union) -- with the following semantics:
\[
M,s \models [\phi\sa_a]\psi \Leftrightarrow \left[M,s \models K_a \phi \Rightarrow M^{\phi\dagger},(s,a) \models \psi\right].\]
This would be completely deterministic, corresponding to exactly one pointed action model instead of the union, and like for general action models and union we would have the following:
\[M,s \models [\phi\sa]\psi \text{ iff } M,s \models \bigwedge_{a \in N} [\phi\sa_a]\psi.\]

A conceptually closely related work is van Ditmarsch's analysis of the \emph{Russian Cards Problem} \cite{van2003russian} which also models \emph{safe announcements} -- albeit with another notion of safety, namely that a secret ``card deal" is not revealed instead of not revealing the identity of the announcer. A safe announcement of $\phi$ by $a$ is captured by $[K_a \phi \wedge [K_a\phi]C \cignorant!]$, where $C$ is common knowledge among all agents and $\cignorant$ is the safety or secrecy condition -- corresponding to only the announcer knowing the identity of the announcer in our case. That is again equivalent to the sequence $[K_a\phi!][C\cignorant!]$, which means that $\phi$ is safe to announce for $a$ if $C\cignorant$ holds after the announcement of $K_a\phi$. Our $[\phi\sai]$ operator also satisfies that definition of safety (Lemma \ref{lemma:ck}). This
raises the question whether a safe anonymous announcement is equivalent to a
pseudo-anonymous announcement followed by an announcement that it is common knowledge that no one except the announcer now knows her identity. We can't refer directly to the identity of ``the announcer" in the updated model in the syntax, but even if we could, the answer to the question would be ``no". Let $\onlyone$ mean that only the announcer now knows her identity. To see that
\[[\phi\sai]\psi \leftrightarrow [\phi\sa][C \onlyone!]\psi\]
does not hold, consider Figure \ref{fig:intentions2}: the $\phi\sa$ update has several more states and $C\onlyone$ would not hold in any of them. One could perhaps consider iterations of the form  $[\phi\sa][\onlyone!][\onlyone!]\cdots$, but that is also not equivalent to $[\phi\sai]$ (for any number of iterations): as a counterexample see Figure \ref{fig:intentions1}. The first announcement of $\onlyone$ would remove the state $(t,a)$, since agent $b$ knows in that state the identity of the announcer. 
The second announcement would remove the state $(t,b)$, but then it stops, resulting in a three-state model. The update with safe announcements however, fails as argued earlier -- $\safe\phi$ does not hold in any state in the original model in Figure \ref{fig:intentions1}.

Regarding safety, we can potentially go even further by using some variant of \textit{distributed knowledge} \cite{vanderhoek99}. That would allow us to reason about anonymous announcements that are safe even when non-announcing agents share their knowledge to deduce the identity of the announcer.

Another avenue of further research is to extend the presented formalisms with \textit{quantification over anonymous announcements}. This will allow us to express properties like ``there is a safe anonymous announcement, such that $\phi$ holds afterwards". Quantification over standard public announcements is relatively well-studied \cite{balbianietal:2008,jal,balbiani2015simple,balbiani2015putting,gal-undec,gald,rustamckjournal}, both in the presence of common and distributed knowledge, and we expect many of the existing intuitions hold also in the case of the anonymous announcements.

Also of interest for future work, is to model \emph{self-referential} intentionally anonymous announcements of the form ``after this very announcement, no-one will know", as studied recently in \cite{baltag_topology_2022,fullfixedpoint}\footnote{That also applies to the Russian Cards problem.}.

\paragraph*{Acknowledgments}
This work was supported by ROIS-DS-JOINT (047RP2024).

\bibliographystyle{splncs04}
\bibliography{anon}

\end{document}